\documentclass[11pt,a4paper]{article}

\usepackage[a4paper,margin=20mm]{geometry}
\usepackage[T1]{fontenc}
\usepackage[utf8]{inputenc}
\usepackage{mathptmx}
\usepackage{graphicx,epsfig}
\usepackage{dsfont,amsfonts,amsbsy,mathrsfs,amscd}
\usepackage{amsmath,amssymb,amsthm}
\usepackage{mathtools}
\usepackage{booktabs}
\usepackage{array}
\usepackage{xcolor}
\usepackage{tikz}
\usepackage{float}
\usepackage{placeins}
\usepackage{microtype}
\usepackage[hidelinks]{hyperref}
\usetikzlibrary{arrows.meta, positioning, shapes.geometric, shapes.misc, fit, backgrounds}

\theoremstyle{plain}
\newtheorem{theorem}{Theorem}[section]
\newtheorem{lemma}[theorem]{Lemma}
\newtheorem{corollary}[theorem]{Corollary}
\theoremstyle{definition}
\newtheorem{definition}[theorem]{Definition}
\newtheorem{example}[theorem]{Example}
\theoremstyle{remark}
\newtheorem{remark}[theorem]{Remark}
\newtheorem*{remark*}{Remark}

\DeclareMathOperator{\sort}{sort}

\DeclareMathOperator{\Decompose}{Decompose}
\DeclareMathOperator{\Inv}{Inv}

\newcommand{\True}{\textnormal{\textsc{true}}}
\newcommand{\False}{\textnormal{\textsc{false}}}
\newcommand{\keywords}[1]{%
  \par\smallskip\noindent\textbf{Keywords:} #1\par\medskip
}

\begin{document}

\title{Multiqubit orthogonal product basis}

\author{Yvkai~Zhao \qquad Lin~Chen\\[0.6em]
\parbox{0.82\textwidth}{\centering\normalsize The authors are with the School of Mathematical Sciences, Beihang University, Beijing 100191, China.}}
\date{}

\maketitle

\begin{abstract}
We use edge-colored complete multigraphs to study complete orthogonal product bases (OPBs) in $n$-qubit systems. We prove that two OPBs are equivalent if and only if their associated multigraphs are isomorphic, thereby reducing OPB classification to graph isomorphism. Within this framework, we establish the upper bound $v\le 2^n-1$ on the number of variables of an $n$-qubit OPB. We also derive $\binom{a_{n-1}+1}{2}\le a_n\le B_{2^{n-1}}^n$ for the number $a_n$ of equivalence classes of $n$-qubit OPBs, where $B_m$ denotes the number of partitions of an $m$-element set. These bounds imply the asymptotic behavior $a_n=2^{2^{n+o(n)}}$. For every OPB, the connectivity pattern of its color layers characterizes local irreducibility, which in turn implies indistinguishability by finite-round local operations and classical communication (LOCC); the existence of a complete color-splitting tree characterizes perfect distinguishability by finite-round LOCC. Finally, we give an algorithm for testing OPB equivalence and a recursive graph algorithm that constructs a finite-round LOCC protocol whenever such perfect discrimination is possible.
\end{abstract}


\tableofcontents
\clearpage

\section{Introduction}
\label{sec:intro}

Quantum entanglement is a central resource in quantum information theory, yet nonlocality without entanglement'' shows that it is not a prerequisite for nonclassical phenomena. Complete bases of pairwise orthogonal product states---\emph{orthogonal product bases} (OPBs)---may contain no entanglement and nevertheless fail to be perfectly distinguishable by local operations and classical communication (LOCC) \cite{ben-nonloc}. Subsequent work studied local distinguishability and nonlocality of orthogonal product states in bipartite and multiqubit systems \cite{feng-shil,lebl,chen-jiang}. OPBs are also closely related to \emph{unextendible product bases} (UPBs): a UPB is an incomplete orthogonal product set whose orthogonal complement contains no product state, and such sets provide a standard construction of bound-entangled states \cite{ben-unext,dmss}. UPBs and OPBs further arise in quantum information hiding, quantum secret sharing, and entanglement-witness constructions.

A systematic classification of OPBs, enumerating their structure up to a natural equivalence, is therefore foundational. This paper concerns complete orthogonal product bases in the multiqubit space.Motivated by the orthogonality structure of OPBs, we study their discrete combinatorial representation through edge-colored complete multigraphs.

Since orthogonality is the essential combinatorial structure of a product basis, graph representations are natural. Orthogonality graphs and related graph-theoretic models have been used to study unextendible product bases, local distinguishability, and isomorphism of nonlocal sets of orthogonal product states \cite{dmss,shi-graph,xu-iso}. Another representation particularly relevant to the present work is the \emph{formal-matrix} framework for multiqubit OPBs \cite{chen-djok}. We use this formalism as the starting point and convert the classification problem into an isomorphism problem for a highly structured class of edge-colored complete multigraphs.

Given an OPB formal matrix, we associate an \emph{edge-colored complete multigraph}: the vertices are the $N=2^n$ row indices, each variable pair contributes a balanced complete bipartite graph, each column contributes one color layer, and the union of the $n$ color layers covers the edges of the complete graph $K_N$. The connected components within a color layer record the variable pairs occurring in the corresponding column, while the bipartition of each component records the two orthogonal directions of that variable pair. Thus the graph retains precisely the combinatorial information relevant to equivalence, while forgetting variable names and choices of orientation. Figure~\ref{fig:example} illustrates this construction for a three-qubit OPB by displaying the individual color layers and their merged edge-colored complete multigraph.

Our main results, presented in the order in which they appear in the paper, are the following.

\begin{enumerate}

\item \textbf{Structural properties and the variable bound.}
We first establish the balancedness property of an OPB formal matrix: whenever a variable pair $\{x_\alpha,x_\alpha^\perp\}$ occurs in a column, the two symbols occur equally many times. Consequently, every variable pair gives a balanced complete bipartite graph $K_{d_\alpha,d_\alpha}$ in the corresponding color layer.

The multigraph structure also yields the variable bound. Orienting the two
sides of each variable component associates every vertex with a
codimension-$n$ subcube of $\{0,1\}^v$. Since the variable components cover
all edges of $K_N$, distinct vertices are separated by some balanced
complete bipartite component, and the resulting subcubes are therefore
pairwise disjoint. They form a tight minimal partition of $\{0,1\}^v$, so
Tarsi's lemma gives
\[
v\le 2^n-1.
\]

\item \textbf{Graph characterization of OPB equivalence.}
Our central structural result is the equivalence theorem (Theorem~\ref{thm:equiv}):
\[
M_1\sim M_2
\quad\Longleftrightarrow\quad
\Gamma(M_1)\cong\Gamma(M_2).
\]
Thus two OPB formal matrices are equivalent exactly when their associated edge-colored complete multigraphs are isomorphic, allowing both vertex and color permutations. The connected components of each layer recover the variable pairs, and the uniqueness of the bipartition of a connected bipartite graph recovers the possible direction flips. Hence the classification of multiqubit OPBs is reduced to an isomorphism problem for this structured family of multigraphs.

\item \textbf{Local reducibility and LOCC distinguishability.}
For a standard-form realization, the connected components of a color layer determine whether the corresponding party can initiate a nontrivial orthogonality-preserving measurement. We prove that $\mathcal B(M)$ is locally irreducible if and only if $\Gamma_j(M)$ is disconnected for every $j$ (Theorem~\ref{thm:local-reducibility} and Corollary~\ref{cor:local-irreducibility}). Local reducibility only determines whether an informative first measurement is possible, so complete discrimination requires a recursive criterion. We define a color-splitting tree and prove that $\mathcal B(M)$ is perfectly distinguishable by finite-round LOCC if and only if $\Gamma(M)$ admits a complete such tree (Theorem~\ref{thm:locc-tree}). Whenever the tree exists, it is both a graph certificate and an explicit protocol of local projective measurements.

Criteria stated directly through admissible local measurements or recursive decompositions of the basis can obscure the common combinatorial structure of successive LOCC rounds \cite{de-rinaldis,lebl}. In the multigraph representation, the first informative measurement is detected by color-layer connectivity, and every subsequent step is obtained by applying the same test to an induced layer. The resulting splitting tree is therefore a compact certificate of distinguishability as well as a transparent description of the adaptive protocol.

\item \textbf{Bounds on the number of OPB classes.}
To quantify the combinatorial structure of the multigraph representation, we
first enumerate individual color layers. For a prescribed spectrum
$\lambda=(d_1,\ldots,d_k)\vdash N/2$, we obtain an exact formula for the
number $L_N(\lambda)$ of labeled layers and a recurrence for the total number
$L_N$ of possible single layers.

We then turn to the number $a_n$ of equivalence classes of $n$-qubit OPBs.
A recursive construction combining two $n$-qubit OPBs gives
\[
a_{n+1}\ge \binom{a_n+1}{2}.
\]
For the upper bound, the multigraph structure induces, in each color layer,
a forced perfect matching of $2^{n-1}$ edges, whose set partitions encode
the possible variable branches. This yields
\[
a_n\le B_{2^{n-1}}^n,
\]
where $B_m$ denotes the $m$-th Bell number.
The forced-matching construction is illustrated explicitly for the three-qubit case in Section~\ref{sec:ub}.

\item \textbf{Asymptotic growth and numerical results.}
Combining the preceding lower bound with the standard asymptotics of the Bell numbers gives
\[
2^{c2^n}
\le
a_n
\le
2^{Cn^22^n}
\]
for suitable constants $c,C>0$. In particular,
\[
\log\log a_n=n+o(n),
\]
or equivalently,
\[
a_n=2^{2^{n+o(n)}}.
\]
Thus the number of essentially different multiqubit OPBs grows doubly exponentially with the number of qubits. Table~\ref{tab:numerics} gives numerical values and bounds for $n=2,\ldots,8$, illustrating the rapid growth predicted by the asymptotic analysis.

\item \textbf{Equivalence-testing algorithm.}
Finally, the multigraph characterization leads to an exact algorithm for deciding whether two OPB formal matrices are equivalent. The first phase compares inexpensive necessary invariants, including the global spectrum, column spectra, and edge-multiplicity data. If these tests are passed, the second phase calls an existing exact graph-isomorphism algorithm on the associated edge-colored complete multigraphs.

We prove that the algorithm returns \True{} if and only if the two input formal matrices are equivalent. We also formulate a recursive graph algorithm for the LOCC criterion: it detects the locally irreducible case from disconnected color layers and records a complete color-splitting tree whenever finite-round LOCC discrimination is possible. At each internal node, the recorded color identifies the party to act, the associated variable pair specifies its two-outcome projective measurement, and the observed outcome selects the next subtree. The full tree thus represents the adaptive LOCC protocol, with each leaf identifying a unique basis state.
\end{enumerate}

The paper is organized as follows. Section~\ref{sec:notation} introduces OPB formal matrices and their edge-colored complete multigraphs, proves the balancedness and variable-bound results, gives an explicit three-qubit example, and establishes the equivalence theorem. Section~\ref{sec:local-locc} characterizes local irreducibility and finite-round LOCC distinguishability in terms of color-layer connectivity and recursive splitting trees. Section~\ref{sec:bounds} develops the enumeration of individual color layers, derives lower and upper bounds on the number of equivalence classes, and obtains the asymptotic growth together with numerical results. Section~\ref{sec:alg} gives graph algorithms for equivalence testing and for constructing a finite-round LOCC discrimination protocol. Section~\ref{sec:concl} concludes the paper and discusses directions for further study.
\section{Edge-colored complete multigraphs and the equivalence theorem}
\label{sec:notation}
\subsection{Notation and OPB formal matrices}
\label{sec:not-mat}
\begin{definition}[OPB formal matrix]
\label{def:fm}
Let $n \ge 1$ and $N=2^n$. Let $\{x_\alpha,x_\alpha^\perp:\alpha\in D\}$ be an alphabet, where $D$ is an index set and $x_\alpha^\perp$ denotes the unit vector orthogonal to $x_\alpha$. An $N\times n$ matrix $M=(M_{rj})_{r\in V,\,1\le j\le n}$ with entries from this alphabet is called an $n$-qubit \emph{OPB formal matrix} if:
\begin{enumerate}
\item each row contains exactly $n$ pairwise distinct letters;
\item any two distinct rows $r\ne s$ are orthogonal: there exist a column $j$ and an index $\alpha$ such that $\{M_{rj},M_{sj}\}=\{x_\alpha,x_\alpha^\perp\}$;
\item realizability: there exist unit vectors $u_\alpha\in\mathbb{C}^2$ with $u_\alpha\perp u_\alpha^\perp$ such that, upon replacing each letter by the corresponding vector, the row vectors form an orthonormal basis of $(\mathbb{C}^2)^{\otimes n}$.
\end{enumerate}
Row $r$ corresponds to the product state $\bigotimes_{j=1}^n \widehat{M}_{rj}$, where $\widehat{M}_{rj}$ is the vector replacing $M_{rj}$. The set of all $n$-qubit OPB formal matrices is denoted by $\mathcal O(n)$.
\end{definition}
\begin{definition}[Equivalence]
\label{def:equiv}
Two OPB formal matrices $M_1,M_2\in \mathcal O(n)$ are \emph{equivalent}, written $M_1\sim M_2$, if $M_2$ is obtained from $M_1$ by the following operations:
\begin{enumerate}
\item[(i)] \emph{row permutations};
\item[(ii)] \emph{column permutations};
\item[(iii)] \emph{variable renamings} $x_\alpha\mapsto x_{\eta(\alpha)}$, where $\eta$ is a bijection of the variable index sets, compatible with the column permutation;
\item[(iv)] \emph{direction flips}: for a set of variables, interchanging $x_\alpha$ and $x_\alpha^\perp$.
\end{enumerate}
These operations are precisely the notational transformations that do not change the physical essence of a product basis: reordering rows, reordering tensor factors, renaming the directions in each qubit, and interchanging the two basis vectors of a qubit direction pair. Equivalence classes therefore correspond to essentially different OPBs.
\end{definition}
Throughout the remainder of the paper, $n\ge1$ is fixed. We write $N=2^n$, $V=\{1,\dots,N\}$, and $D=\{1,\dots,v\}$, where $v$ is called the number of \emph{free variables}. By Definition~\ref{def:fm}, an $n$-qubit OPB formal matrix is written as
$M=(M_{rj})_{r\in V,\,1\le j\le n}\in \mathcal O(n)$,
and row $r$ represents a product state. We only consider variables that actually occur in the matrix.
Each variable pair $\{x_\alpha,x_\alpha^\perp\}$ occurs only in a fixed column; we denote this column by $\chi(\alpha)\in\{1,\dots,n\}$.

We adopt the \emph{standard form} of OPB formal matrices: distinct variable pairs in the same column have distinct direction pairs, i.e., the realizing vectors $u_\alpha$ and $u_\beta$ of two distinct variables $x_\alpha\ne x_\beta$ in the same column are neither parallel nor orthogonal. For each variable $\alpha\in D$, define
\[
\begin{gathered}
P_\alpha(M)=\{r\in V:M_{r,\chi(\alpha)}=x_\alpha\},\\
N_\alpha(M)=\{r\in V:M_{r,\chi(\alpha)}=x_\alpha^\perp\}.
\end{gathered}
\]
Write $d_\alpha=|P_\alpha(M)|$; by Lemma~\ref{lem:balance},
$|P_\alpha(M)|=|N_\alpha(M)|=d_\alpha>0$.
Arranging the part sizes in nonincreasing order gives the \emph{spectrum} of $M$,
$\pi(M)=(d_1,\dots,d_v)$.
For each column $j$, let
$A_j=\{\alpha\in D:\chi(\alpha)=j\}$.
Since column $j$ has $N$ entries, the family
$\{P_\alpha(M)\sqcup N_\alpha(M):\alpha\in A_j\}$
partitions $V$, and
$\sum_{\alpha\in A_j}d_\alpha=2^{n-1}=\frac{N}{2}$.
\subsection{Edge-colored complete multigraphs}
\label{sec:col-graph}
\begin{definition}[Edge-colored complete multigraph]
\label{def:gamma}
Given $M\in \mathcal O(n)$, construct the edge-colored complete multigraph $\Gamma(M)$ as follows:
\begin{itemize}
\item the vertex set is $V$;
\item for each variable $\alpha\in D$ and each $r\in P_\alpha(M)$, $s\in N_\alpha(M)$, the pair $\{r,s\}$ is joined by an edge of color $\chi(\alpha)$;
\item equivalently, column $j$ contributes the \emph{$j$-th color layer},
\[
\Gamma_j(M)=\bigsqcup_{\alpha\in A_j}K(P_\alpha(M),N_\alpha(M));
\]
\item merging the $n$ color layers, with multi-edges allowed, yields the edge-colored complete multigraph on $V$.
\end{itemize}
By condition~(2) of Definition~\ref{def:fm}, any two distinct vertices are adjacent in some layer; conversely, two adjacent vertices in the same layer correspond to orthogonal rows. Therefore the $n$ color layers cover exactly the edges of the complete graph $K_N$:
\[
E(K_N)
=
\bigcup_{\alpha=1}^{v}E\bigl(K(P_\alpha(M),N_\alpha(M))\bigr)
=
\bigcup_{j=1}^{n}E(\Gamma_j(M)).
\]
For an unordered pair $\{r,s\}\subseteq V$ with $r\ne s$ and a color $j$, write $\mu_M(\{r,s\},j)\in\{0,1\}$ for the multiplicity of the edge $\{r,s\}$ in the color layer $\Gamma_j(M)$. Each layer is a simple graph, so the multiplicity is $0$ or $1$.
\end{definition}
\begin{definition}[Equivalence of edge-colored complete multigraphs]
\label{def:gamma-equiv}
Two edge-colored complete multigraphs $\Gamma(M_1)$ and $\Gamma(M_2)$ are \emph{equivalent}, written $\Gamma(M_1)\cong\Gamma(M_2)$, if there exist a vertex permutation $\sigma:V\to V$ and a color permutation $\tau:\{1,\dots,n\}\to\{1,\dots,n\}$ such that for every pair $\{r,s\}$ and every color $j$,
\[
\mu_{M_1}(\{r,s\},j)
=
\mu_{M_2}(\{\sigma(r),\sigma(s)\},\tau(j)).
\]
\end{definition}
\begin{remark}
\label{rem:ops}
The four equivalence operations on OPB formal matrices correspond to transformations of the edge-colored complete multigraphs as follows:
\begin{center}
\resizebox{\columnwidth}{!}{%
\begin{tabular}{cc}
\toprule
operation on $M_1\sim M_2$ & corresponding transformation of $\Gamma(M_1)\cong\Gamma(M_2)$ \\
\midrule
row permutation & vertex permutation \\
column permutation & color permutation \\
variable renaming & none (implicit) \\
direction flip $x_\alpha\leftrightarrow x_\alpha^\perp$ & none (implicit) \\
\bottomrule
\end{tabular}%
}
\end{center}
Variable renamings and direction flips do not alter any edge structure of the graph: renaming changes only the names of variables, and a direction flip merely interchanges the two parts of a bipartite graph. Hence they generate no independent graph operations. This observation is the intuitive basis of the equivalence theorem in Section~\ref{sec:equiv}.
\end{remark}
\subsection{Balancedness}
\label{sec:balance}
\begin{lemma}[Balancedness]
\label{lem:balance}
For every formal matrix $M\in \mathcal O(n)$, every column $j$, and every letter $x_\alpha$ appearing in column $j$, one has $|P_\alpha(M)|=|N_\alpha(M)|>0$.
\end{lemma}
\begin{proof}
Since an OPB is a complete basis, the number of rows is $2^n$. Fix a letter $x_\alpha$ in column $j$ and set
\[
P_\alpha=\{r\in V:M_{rj}=x_\alpha\},
\qquad
N_\alpha=\{r\in V:M_{rj}=x_\alpha^\perp\},
\]
and
$C=\{r\in V:M_{rj}\ne x_\alpha \text{ and } M_{rj}\ne x_\alpha^\perp\}$.
Write $a=|P_\alpha|$ and $b=|N_\alpha|$.
Since the rows in $P_\alpha$ all have $x_\alpha$ as their $j$-th component, the full inner product of any two of them equals the inner product of the rows with the $j$-th component deleted, multiplied by $\langle x_\alpha,x_\alpha\rangle$. By orthogonality of the OPB, the rows of $P_\alpha$ remain pairwise orthogonal after deleting the $j$-th component, so they span an $a$-dimensional subspace $U$. 

Similarly, the rows of $N_\alpha$ are pairwise orthogonal after deleting the $j$-th component and span a $b$-dimensional subspace $T$. Since the $j$-th components of $P_\alpha$ and $N_\alpha$ are $x_\alpha$ and $x_\alpha^\perp$, respectively, $P_\alpha\cap N_\alpha=\varnothing$.
Column permutations of an OPB formal matrix do not affect balancedness, so we may move column $j$ to the last position. The rows in $P_\alpha$ can be written as $u\otimes x_\alpha$ with $u\in U$, and the rows in $N_\alpha$ as $w\otimes x_\alpha^\perp$ with $w\in T$. The tensor product of the remaining $n-1$ factor spaces has dimension $2^{n-1}$. Put $r=\dim(U+T)$; then
$\dim(U+T)^\perp=2^{n-1}-r$.

Consider any row $w_1\otimes\cdots\otimes w_{n-1}\otimes y$ of $C$, and write $w$ for its truncation to the first $n-1$ components. Since $y\ne x_\alpha,x_\alpha^\perp$, by the standard-form convention of Section~\ref{sec:not-mat}, the vector $y$ is neither parallel nor orthogonal to $x_\alpha$, so
\[
\langle y,x_\alpha\rangle\ne0
\qquad\text{and}\qquad
\langle y,x_\alpha^\perp\rangle\ne0.
\]
For this row to be orthogonal to all rows of $P_\alpha$ and $N_\alpha$, one must have $w\perp U$ and $w\perp T$, i.e., $w\in(U+T)^\perp$. Hence all rows of $C$ lie in $(U+T)^\perp\otimes H_n$, a space of dimension $2(2^{n-1}-r)$. Since the rows in $C$ must be pairwise orthogonal,
$|C|\le 2(2^{n-1}-r)$.

On the other hand,
$2^n = |P_\alpha|+|N_\alpha|+|C| = a+b+|C| \le a+b+2^n-2r$,
so $a+b\ge2r$. Since $r=\dim(U+T)$, we have $r\ge a$ and $r\ge b$, hence $2r\ge a+b$. Therefore $a+b=2r$, and equality gives $a=b=r$.
We only consider variables that actually occur, so $|P_\alpha|>0$ and $|N_\alpha|>0$. Hence
$|P_\alpha|=|N_\alpha|>0$ and $P_\alpha\cap N_\alpha=\varnothing$.
\end{proof}
\subsection{The bound on the number of variables}
\label{sec:varbound}
\begin{theorem}[Bound on the number of variables]
\label{thm:var}
For every $M\in \mathcal O(n)$, the number of free variables satisfies $v\le2^n-1$.
\end{theorem}
This conclusion is not new: an equivalent bound was obtained via color forests of admissible colorings of the hypercube \cite{lebl}. To unify the language of this paper, we give another proof in the present multigraph framework.
\begin{proof}
For each variable pair $\{x_\alpha,x_\alpha^\perp\}$, choose one of its two directions and call it $0$ and the other $1$. For row $r$, define a partial assignment
$\rho_r:\{1,\dots,v\}\to\{0,1,\ast\}$
as follows:
\begin{enumerate}
\item if row $r$ contains $x_\alpha$, set $\rho_r(\alpha)=0$;
\item if row $r$ contains $x_\alpha^\perp$, set $\rho_r(\alpha)=1$;
\item if row $r$ does not contain the variable pair $\alpha$, set $\rho_r(\alpha)=\ast$.
\end{enumerate}
Since every row contains exactly $n$ pairwise distinct letters, $\rho_r$ fixes exactly $n$ variable coordinates. Let
\[
Q_r
=
\{z\in\{0,1\}^v:z_\alpha=\rho_r(\alpha)
\text{ whenever }\rho_r(\alpha)\ne\ast\}.
\]
Then $Q_r$ is a subcube of $\{0,1\}^v$ of codimension $n$, and
$|Q_r|=2^{v-n}$.
Take two rows $r\ne s$. By the covering condition
$E(K_N)=\bigcup_{\alpha=1}^{v}E(K(P_\alpha,N_\alpha))$,
there exists a variable $x_\alpha$ with $\{r,s\}\in E(K(P_\alpha,N_\alpha))$. Thus $r$ and $s$ lie on opposite sides of the complete bipartite graph. The two row assignments therefore fix $z_\alpha=0$ and $z_\alpha=1$, respectively, so $Q_r$ and $Q_s$ conflict on coordinate $\alpha$ and
$Q_r\cap Q_s=\varnothing$.
Therefore $Q_1,\dots,Q_N$ are pairwise disjoint subcubes, and
\[
\left|\bigcup_{r=1}^{N}Q_r\right|
=
\sum_{r=1}^{N}|Q_r|
=
N\cdot2^{v-n}
=
2^v.
\]
Since $\bigcup_{r=1}^{N}Q_r\subseteq\{0,1\}^v$ and $|\{0,1\}^v|=2^v$, necessarily
$\bigcup_{r=1}^{N}Q_r=\{0,1\}^v$.

Thus $\mathcal{Q}(M)=\{Q_1,\dots,Q_N\}$ is a \emph{subcube partition} of $\{0,1\}^v$.
We claim that $\mathcal{Q}(M)$ is a \emph{tight minimal cover}. It is a minimal cover because it is a partition, so deleting any $Q_r$ leaves the points it contains uncovered. It is tight because every variable occurs in the matrix; hence each Boolean coordinate is fixed by at least one subcube.

By Tarsi's lemma \cite{aharoni,firmus}, if $\{0,1\}^{m}$ admits a tight minimal cover consisting of $L$ subcubes, then $L\ge m+1$. Here $m=v$ and $L=N=2^n$, so
$2^n=N\ge v+1$,
and therefore
$v\le 2^n-1$.
\end{proof}
\subsection{An example}
\label{sec:example}
Consider the three-qubit OPB formal matrix
\[
M=
\begin{pmatrix}
x_1 & x_2 & x_3 \\
x_1 & x_2 & x_3^\perp \\
x_1 & x_2^\perp & x_3 \\
x_1 & x_2^\perp & x_3^\perp \\
x_1^\perp & x_2 & x_4 \\
x_1^\perp & x_2 & x_4^\perp \\
x_1^\perp & x_2^\perp & x_4 \\
x_1^\perp & x_2^\perp & x_4^\perp
\end{pmatrix}.
\]
Taking $x_1=x_2=x_3=x_4=|0\rangle$ yields the computational basis $\{|000\rangle,\dots,|111\rangle\}$, so this is a legitimate three-qubit OPB formal matrix. Since the $x_\alpha$ are formal symbols, the realizing states are a choice: taking $x_4=|+\rangle$ instead of $|0\rangle$ still realizes an orthonormal basis of $\mathbb{C}^8$, and the two variable pairs in column $3$ then have distinct direction pairs, as the standard form of Section~\ref{sec:not-mat} requires.
Its edge-colored complete multigraph consists of three color layers:
\begin{itemize}
\item color $1$ (variable $x_1$, column $1$): $K(P_1,N_1)$, where $P_1=\{1,2,3,4\}$ and $N_1=\{5,6,7,8\}$;
\item color $2$ (variable $x_2$, column $2$): $K(P_2,N_2)$, where $P_2=\{1,2,5,6\}$ and $N_2=\{3,4,7,8\}$;
\item color $3$ (variables $x_3,x_4$, column $3$):
$K(P_3,N_3)\sqcup K(P_4,N_4)$,
where $P_3=\{1,3\}$, $N_3=\{2,4\}$, $P_4=\{5,7\}$, and $N_4=\{6,8\}$.
\end{itemize}
Merging the three color layers yields the edge-colored complete multigraph on $\{1,\dots,8\}$; see Fig.~\ref{fig:example}.
\definecolor{layerblue}{HTML}{3F73E8}
\definecolor{layerorange}{HTML}{E58A00}
\definecolor{layergreen}{HTML}{24A865}
\tikzset{
  vertex/.style={
    circle, draw=black, fill=white, line width=.45pt,
    inner sep=.2pt, minimum size=4.0mm,
    font=\tiny\bfseries
  },
  layerone/.style={draw=layerblue, line width=.8pt, line cap=round},
  layertwo/.style={draw=layerorange, line width=.8pt, line cap=round},
  layerthree/.style={draw=layergreen, line width=.9pt, line cap=round},
  mergedone/.style={draw=layerblue, line width=.55pt, opacity=.72},
  mergedtwo/.style={draw=layerorange, dashed, line width=.65pt, opacity=.80},
  mergedthree/.style={draw=layergreen, dotted, line width=.95pt, opacity=.92},
  layerlabel/.style={font=\small\bfseries},
  graphlabel/.style={font=\scriptsize, text=black!70}
}
\newcommand{\LayerCoordinates}[1]{%
  \coordinate (#1p1) at (0,-1.35);%
  \coordinate (#1p2) at (-1.20,-.78);%
  \coordinate (#1p3) at (-1.62,.25);%
  \coordinate (#1p4) at (-1.05,1.25);%
  \coordinate (#1p5) at (0,1.55);%
  \coordinate (#1p6) at (1.05,1.25);%
  \coordinate (#1p7) at (1.62,.25);%
  \coordinate (#1p8) at (1.20,-.78);%
}
\newcommand{\LayerVertices}[1]{%
  \node[vertex] (#1v1) at (#1p1) {$v_1$};%
  \node[vertex] (#1v2) at (#1p2) {$v_2$};%
  \node[vertex] (#1v3) at (#1p3) {$v_3$};%
  \node[vertex] (#1v4) at (#1p4) {$v_4$};%
  \node[vertex] (#1v5) at (#1p5) {$v_5$};%
  \node[vertex] (#1v6) at (#1p6) {$v_6$};%
  \node[vertex] (#1v7) at (#1p7) {$v_7$};%
  \node[vertex] (#1v8) at (#1p8) {$v_8$};%
}
\newcommand{\LayerOneEdges}[2]{%
  \foreach \u in {#1v1,#1v2,#1v3,#1v4}{%
    \foreach \v in {#1v5,#1v6,#1v7,#1v8}{%
      \draw[#2] (\u)--(\v);%
    }%
  }%
}
\newcommand{\LayerTwoEdges}[2]{%
  \foreach \u in {#1v1,#1v2,#1v5,#1v6}{%
    \foreach \v in {#1v3,#1v4,#1v7,#1v8}{%
      \draw[#2] (\u)--(\v);%
    }%
  }%
}
\newcommand{\LayerThreeEdges}[2]{%
  \foreach \u/\v in {
    #1v1/#1v2,#1v1/#1v4,#1v3/#1v2,#1v3/#1v4,
    #1v5/#1v6,#1v5/#1v8,#1v7/#1v6,#1v7/#1v8%
  }{%
    \draw[#2] (\u)--(\v);%
  }%
}
\newcommand{\MergedLayerOne}[1]{%
  \foreach \u in {#1v1,#1v2,#1v3,#1v4}{%
    \foreach \v in {#1v5,#1v6,#1v7,#1v8}{%
      \draw[mergedone] (\u)--(\v);%
    }%
  }%
}
\newcommand{\MergedLayerTwo}[1]{%
  \foreach \u in {#1v1,#1v2,#1v5,#1v6}{%
    \foreach \v in {#1v3,#1v4,#1v7,#1v8}{%
      \draw[mergedtwo] (\u)--(\v);%
    }%
  }%
}
\newcommand{\MergedLayerThree}[1]{%
  \foreach \u/\v in {
    #1v1/#1v2,#1v1/#1v4,#1v3/#1v2,#1v3/#1v4,
    #1v5/#1v6,#1v5/#1v8,#1v7/#1v6,#1v7/#1v8%
  }{%
    \draw[mergedthree] (\u)--(\v);%
  }%
}
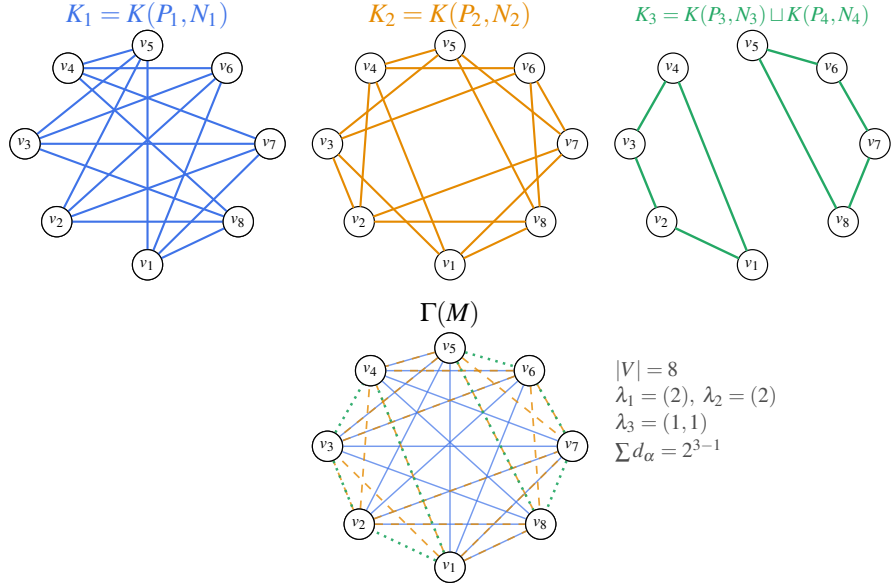
\begin{figure}[t]
\centering
\begin{tikzpicture}[x=1cm,y=1cm]
\begin{scope}[shift={(0,0)}]
  \LayerCoordinates{a}
  \LayerVertices{a}
  \LayerOneEdges{a}{layerone}
  \LayerVertices{a}
  \node[layerlabel, text=layerblue] at (0,1.95)
    {$K_1=K(P_1,N_1)$};
\end{scope}
\begin{scope}[shift={(4.00,0)}]
  \LayerCoordinates{b}
  \LayerVertices{b}
  \LayerTwoEdges{b}{layertwo}
  \node[layerlabel, text=layerorange] at (0,1.95)
    {$K_2=K(P_2,N_2)$};
\end{scope}
\begin{scope}[shift={(8.00,0)}]
  \LayerCoordinates{c}
  \LayerVertices{c}
  \LayerThreeEdges{c}{layerthree}
  \node[font=\scriptsize\bfseries, text=layergreen, align=center]
    at (0,1.95)
    {$K_3=K(P_3,N_3)\sqcup K(P_4,N_4)$};
\end{scope}
\begin{scope}[shift={(4.00,-4.00)}]
  \node[layerlabel] at (0,2.00) {$\Gamma(M)$};
  \LayerCoordinates{d}
  \LayerVertices{d}
  \MergedLayerOne{d}
  \MergedLayerTwo{d}
  \MergedLayerThree{d}
  \LayerVertices{d}
  \node[graphlabel, anchor=west] at (2.05,1.25) {$|V|=8$};
  \node[graphlabel, anchor=west] at (2.05,.90)
    {$\lambda_1=(2),\ \lambda_2=(2)$};
  \node[graphlabel, anchor=west] at (2.05,.55)
    {$\lambda_3=(1,1)$};
  \node[graphlabel, anchor=west] at (2.05,.20)
    {$\sum d_\alpha=2^{3-1}$};
\end{scope}
\end{tikzpicture}
\caption{The three color layers of the example and the merged edge-colored complete multigraph $\Gamma(M)$ on $V=\{1,\dots,8\}$. The three layers are drawn in blue (solid), orange (dashed), and green (dotted), respectively.}
\label{fig:example}
\end{figure}

The spectra of the three layers are $\lambda_1=(2)$, $\lambda_2=(2)$, and $\lambda_3=(1,1)$. For each layer,
$\sum_{\alpha\in A_j}d_\alpha=2^{n-1}$,
in accordance with Lemma~\ref{lem:balance}.
\subsection{The equivalence theorem}
\label{sec:equiv}
\begin{theorem}
\label{thm:equiv}
Let $M_1,M_2\in \mathcal O(n)$ be two OPB formal matrices. Then $M_1\sim M_2$ if and only if $\Gamma(M_1)\cong\Gamma(M_2)$.
\end{theorem}
\begin{proof}
$(\Rightarrow)$ Let $\eta$ be the variable renaming, $\tau$ the column permutation, and $\varepsilon_\alpha\in\{0,1\}$ the flip indicators, with
$\chi_{M_2}(\eta(\alpha))=\tau(\chi_{M_1}(\alpha))$.
If $\varepsilon_\alpha=0$, then
\[
\sigma(P_\alpha(M_1))=P_{\eta(\alpha)}(M_2),
\qquad
\sigma(N_\alpha(M_1))=N_{\eta(\alpha)}(M_2);
\]
if $\varepsilon_\alpha=1$, the two sides are interchanged. An edge $\{r,s\}$ of color $j=\chi_{M_1}(\alpha)$ in $\Gamma(M_1)$ satisfies $r\in P_\alpha(M_1)$ and $s\in N_\alpha(M_1)$. In either case, $\sigma(r)$ and $\sigma(s)$ lie on opposite sides of
$K(P_{\eta(\alpha)}(M_2),N_{\eta(\alpha)}(M_2))$,
hence $\{\sigma(r),\sigma(s)\}$ is an edge of color
$\chi_{M_2}(\eta(\alpha))=\tau(j)$.
Since all mappings are bijections, the same holds for the inverse, and therefore
$$\mu_{M_1}(\{r,s\},j) = \mu_{M_2}(\{\sigma(r),\sigma(s)\},\tau(j))$$
for all pairs $\{r,s\}$ and colors $j$. Thus $\Gamma(M_1)\cong\Gamma(M_2)$.

$(\Leftarrow)$ Conversely, assume $\Gamma(M_1)\cong\Gamma(M_2)$ via vertex and color permutations $\sigma,\tau$. The row and column permutations are $\sigma$ and $\tau$; it remains to recover the renaming and the flips.
Fix $M$ and $j$. The blocks
$K(P_\alpha(M),N_\alpha(M))$, $\alpha\in A_j(M)$,
have pairwise disjoint vertex sets, and each is connected. Hence the connected components of $\Gamma_j(M)$ are exactly these blocks: the layer loses the variable names but not the variable blocks.
Since $\sigma$ maps $\Gamma_j(M_1)$ isomorphically onto $\Gamma_{\tau(j)}(M_2)$, it maps components to components. For each $\alpha\in A_j(M_1)$ there is thus a unique $\eta(\alpha)\in A_{\tau(j)}(M_2)$ with
\[
\sigma(P_\alpha(M_1)\sqcup N_\alpha(M_1))
=
P_{\eta(\alpha)}(M_2)\sqcup N_{\eta(\alpha)}(M_2).
\]
As $\sigma$ carries the components of color $j$ bijectively to those of color $\tau(j)$, $\eta$ is a bijection of the variable sets, and
$\chi_{M_2}(\eta(\alpha)) = \tau(\chi_{M_1}(\alpha))$
holds automatically.
It remains to recover the directions. Fix $\alpha$, and put
\[
\begin{gathered}
H=K(P_\alpha(M_1),N_\alpha(M_1)),\\
H'=K(P_{\eta(\alpha)}(M_2),N_{\eta(\alpha)}(M_2)).
\end{gathered}
\]
The bipartition of a connected bipartite graph is unique up to swapping the two sides. Hence $\sigma(P_\alpha(M_1))$ is either $P_{\eta(\alpha)}(M_2)$ or $N_{\eta(\alpha)}(M_2)$; the former corresponds to no flip, and the latter to the flip
$x_{\eta(\alpha)}\leftrightarrow x_{\eta(\alpha)}^\perp$.
Finally, take any position $(r,j)$ of $M_1$. The supports of column $j$ partition $V$, so there is a unique $\alpha\in A_j(M_1)$ with
$r\in P_\alpha(M_1)\sqcup N_\alpha(M_1)$.
Under $\sigma,\tau$, this position maps to $(\sigma(r),\tau(j))$, and the preceding component correspondence and flip choice reproduce exactly the entry of $M_2$ at that position. Since $(r,j)$ is arbitrary, $M_1\sim M_2$.
\end{proof}

\section{Local reducibility and LOCC distinguishability}
\label{sec:local-locc}

The color layers also encode operational properties of a multiqubit OPB. Existing approaches express distinguishability through local reducibility and admissible measurement sequences \cite{de-rinaldis,chen-jiang}, through structural representations in fixed low-dimensional systems \cite{feng-shil}, or, for multiqubit unentangled orthonormal bases, through recursive decompositions after the first local measurement \cite{lebl}. The multigraph formulation consolidates the multiqubit criterion into two elementary operations on $\Gamma(M)$: testing connectivity of a color layer and repeating that test on the induced layers created by a split. The same finite combinatorial object therefore detects local irreducibility, certifies finite-round LOCC distinguishability, and, when successful, records the adaptive projective protocol. Because these criteria depend only on $\Gamma(M)$, they are invariants of the formal equivalence class by Theorem~\ref{thm:equiv}.

Fix $M\in \mathcal O(n)$ and a standard-form realization
\[
\mathcal B(M)=\left\{
|\psi_r\rangle=\bigotimes_{k=1}^{n}|u_{rk}\rangle:r\in V
\right\}.
\]
For a fixed party $j$, write
$|\psi_r\rangle=|u_{rj}\rangle\otimes|\widehat\psi_r^{(j)}\rangle$.

\begin{definition}[Local reducibility]
\label{def:local-reducibility}
The basis $\mathcal B(M)$ is \emph{$j$-reducible} if $V$ has a partition into two nonempty sets $V_0$ and $V_1$ such that
\[
\operatorname{span}\{|u_{rj}\rangle:r\in V_0\}
\perp
\operatorname{span}\{|u_{sj}\rangle:s\in V_1\}.
\]
It is \emph{locally reducible} if it is $j$-reducible for some $j$, and \emph{locally irreducible} otherwise.

A positive operator $E$ on party $j$ is \emph{orthogonality preserving} for $\mathcal B(M)$ if
\[
\langle\psi_r|(E_j\otimes I_{\bar j})|\psi_s\rangle=0
\qquad(r\ne s).
\]
A local measurement is orthogonality preserving if each of its effects has this property.  It is \emph{trivial} if all its effects are scalar multiples of the identity.
\end{definition}

The first definition agrees with the usual $A_j$-reducibility of an OPB \cite{chen-jiang}.  For complete qubit OPBs, it is also equivalent to the existence of a nontrivial orthogonality-preserving measurement, as the next lemma and theorem show.

\begin{lemma}[Local eigenvector lemma]
\label{lem:local-eigenvector}
If $E$ is an orthogonality-preserving effect on party $j$, then every local factor $|u_{rj}\rangle$ is an eigenvector of $E$.
\end{lemma}

\begin{proof}
Fix $r\in V$.  Orthogonality preservation gives, for every $s\ne r$,
\[
\left\langle\psi_s\middle|
(E|u_{rj}\rangle)\otimes|\widehat\psi_r^{(j)}\rangle
\right\rangle=0.
\]
The displayed vector is therefore orthogonal to all members of the complete orthonormal basis $\mathcal B(M)$ except possibly $|\psi_r\rangle$.  It must be a scalar multiple of $|\psi_r\rangle$.  Since $|\widehat\psi_r^{(j)}\rangle$ is nonzero, there is a scalar $\lambda_r$ such that
\[
E|u_{rj}\rangle=\lambda_r|u_{rj}\rangle.
\]
\end{proof}

\begin{theorem}[Graph criterion for local reducibility]
\label{thm:local-reducibility}
Fix a party $j$ and let $c_j=c(\Gamma_j(M))=|A_j|$.  The following statements are equivalent:
\begin{enumerate}
\item $\mathcal B(M)$ is $j$-reducible;
\item party $j$ admits a nontrivial orthogonality-preserving local measurement;
\item $c_j=1$;
\item the color layer $\Gamma_j(M)$ is connected.
\end{enumerate}
If these conditions fail, every orthogonality-preserving effect on party $j$ is a scalar multiple of the identity.
\end{theorem}

\begin{proof}
If $c_j=1$, all local factors in column $j$ belong to one orthonormal pair $\{u_\alpha,u_\alpha^\perp\}$.  The partition $P_\alpha(M)\sqcup N_\alpha(M)$ proves $j$-reducibility, and the projective measurement
\[
\left\{
|u_\alpha\rangle\langle u_\alpha|,
|u_\alpha^\perp\rangle\langle u_\alpha^\perp|
\right\}
\]
is nontrivial and orthogonality preserving.

Conversely, a $j$-reducing partition consists, in the two-dimensional local space, of two one-dimensional mutually orthogonal spans.  Hence every local factor in column $j$ belongs to the same orthonormal pair.  The standard-form convention then gives $c_j=1$.

Suppose now that $c_j\ge2$ and choose distinct $\alpha,\beta\in A_j$.  By standard form, $u_\alpha$ and $u_\beta$ are neither parallel nor orthogonal.  Lemma~\ref{lem:local-eigenvector} makes both vectors eigenvectors of any orthogonality-preserving effect $E$.  Nonorthogonal eigenvectors of a Hermitian operator have the same eigenvalue; since these two vectors are linearly independent, $E$ is a scalar multiple of the identity.  Thus no nontrivial orthogonality-preserving local measurement exists.

The equivalence of the last two statements follows from
\[
\Gamma_j(M)=
\bigsqcup_{\alpha\in A_j}K(P_\alpha(M),N_\alpha(M)),
\]
whose components are nonempty by Lemma~\ref{lem:balance}.
\end{proof}

\begin{corollary}[Local irreducibility]
\label{cor:local-irreducibility}
For a standard-form multiqubit OPB,
\[
\begin{aligned}
\mathcal B(M)\text{ is locally irreducible}
&\quad\Longleftrightarrow\quad
c(\Gamma_j(M))\ge2\ \text{for every }j\\
&\quad\Longleftrightarrow\quad
\Gamma_j(M)\text{ is disconnected for every }j.
\end{aligned}
\]
In this case $\mathcal B(M)$ cannot be perfectly distinguished by any finite-round LOCC protocol.
\end{corollary}

\begin{proof}
The equivalences follow by applying Theorem~\ref{thm:local-reducibility} to every party.  In a perfect finite-round discrimination protocol, the postmeasurement states at every nonzero outcome must remain pairwise orthogonal.  Measurements whose effects are scalar multiples of the identity reveal no information, and their conditional local unitaries can be absorbed into later operations.  The first remaining measurement is therefore a nontrivial orthogonality-preserving measurement by one of the parties, contradicting local irreducibility.
\end{proof}

Local reducibility is only a condition on the first informative measurement.  A basis may admit such a measurement while one of the resulting branches admits no continuation.  The complete LOCC criterion consequently has to be recursive.

For $U\subseteq V$, let $\Gamma_j(M)[U]$ denote the subgraph of the $j$th layer induced by $U$.  A color $j$ is \emph{splittable on $U$} if $\Gamma_j(M)[U]$ is connected.  When $|U|>1$, this means that there is a unique $\alpha=\alpha(j,U)\in A_j$ such that
\[
U\subseteq P_\alpha(M)\sqcup N_\alpha(M),
\]
and both intersections with $U$ are nonempty.

\begin{definition}[Color-splitting tree]
\label{def:splitting-tree}
A \emph{color-splitting tree} for $M$ is a rooted binary tree whose nodes are labeled by pairs $(U,J)$ with $U\subseteq V$ and $J\subseteq\{1,\ldots,n\}$.  The root is $(V,\{1,\ldots,n\})$.  A node with $|U|=1$ is a leaf.  Every nonleaf node chooses a color $j\in J$ that is splittable on $U$ and has the two children
\[
\left(U\cap P_{\alpha(j,U)}(M),J\setminus\{j\}\right),
\qquad
\left(U\cap N_{\alpha(j,U)}(M),J\setminus\{j\}\right).
\]
The tree is \emph{complete} if all its leaves are singletons.
\end{definition}

Equivalently, define a Boolean recursion $\mathsf D(U,J)$.  Set $\mathsf D(U,J)=\True$ when $|U|=1$.  For $|U|>1$, set
\begin{equation}
\label{eq:dist-recursion}
\mathsf D(U,J)=
\bigvee_{\substack{j\in J:\,\Gamma_j(M)[U]\text{ connected}}}
\left(
\mathsf D(U\cap P_{\alpha(j,U)},J\setminus\{j\})
\wedge
\mathsf D(U\cap N_{\alpha(j,U)},J\setminus\{j\})
\right),
\end{equation}
where an empty disjunction is $\False$.  A complete color-splitting tree exists exactly when
\[
\mathsf D(V,\{1,\ldots,n\})=\True.
\]

\begin{lemma}[Subproblem invariant]
\label{lem:subproblem-invariant}
At every node $(U,J)$ generated from the root by the splitting rule, the factors already removed are fixed on that branch, and the rows indexed by $U$, restricted to the parties in $J$, form a complete OPB of $(\mathbb C^2)^{\otimes |J|}$.  In particular, $|U|=2^{|J|}$.
\end{lemma}

\begin{proof}
The assertion holds at the root.  Suppose it holds at $(U,J)$ and color $j$ is used for the next split.  All states in $U$ use the pair $\{u_\alpha,u_\alpha^\perp\}$ at party $j$.  The states on either side of the split remain pairwise orthogonal after deleting that fixed local factor, so each side contains at most $2^{|J|-1}$ states.  Their sizes sum to $|U|=2^{|J|}$; hence both sides have size $2^{|J|-1}$ and each restricted orthogonal set is complete.  This proves the assertion for both children.
\end{proof}

\begin{theorem}[Recursive graph criterion for LOCC]
\label{thm:locc-tree}
For a standard-form multiqubit OPB $\mathcal B(M)$, the following statements are equivalent:
\begin{enumerate}
\item $\mathcal B(M)$ is perfectly distinguishable by a finite-round LOCC protocol;
\item $\Gamma(M)$ admits a complete color-splitting tree;
\item $\mathsf D(V,\{1,\ldots,n\})=\True$ in \eqref{eq:dist-recursion}.
\end{enumerate}
Whenever these conditions hold, the splitting tree specifies a protocol using only local projective measurements.
\end{theorem}

\begin{proof}
Assume first that a complete color-splitting tree is given.  At an internal node $(U,J)$ labeled by color $j$ and variable $\alpha(j,U)$, party $j$ measures in the basis $\{u_\alpha,u_\alpha^\perp\}$.  The two outcomes restrict the candidates to the two children of the node.  Repeating this operation along the tree reaches a singleton at every leaf, so the accumulated outcomes identify the initial state.  The depth is at most $n$, and the protocol is a finite-round LOCC protocol.

Conversely, suppose that a finite-round LOCC protocol perfectly distinguishes the states at a node $(U,J)$.  By Lemma~\ref{lem:subproblem-invariant}, they constitute a complete OPB on the active parties.  Remove all initial measurements whose effects are scalar multiples of the identity, absorbing their conditional local unitaries into later operations, and let party $j$ perform the first remaining measurement.  Every effect of this measurement must preserve orthogonality: two nonorthogonal conditional pure states cannot subsequently be discriminated with certainty.  Applying Lemma~\ref{lem:local-eigenvector} to the complete subproblem and repeating the standard-form argument in Theorem~\ref{thm:local-reducibility}, a nontrivial effect can exist only when all local factors at party $j$ belong to one variable pair.  Equivalently, $\Gamma_j(M)[U]$ is connected.

Party $j$ may therefore be measured projectively in that variable basis.  Each outcome gives one of the two child subproblems in Definition~\ref{def:splitting-tree}.  The original discrimination protocol, restricted to either promised subset of states, shows that both children remain LOCC distinguishable.  Induction on $|J|$ produces complete splitting trees below both children and hence below $(U,J)$.  Applied at the root, this proves necessity.  The equivalence with the Boolean recursion is immediate from its definition.
\end{proof}

Using the multigraph representation, the theorem distinguishes three operational regimes.  If every layer is disconnected at the root, Corollary~\ref{cor:local-irreducibility} gives local irreducibility and LOCC indistinguishability.  If some root layer is connected but the recursion returns $\False$, the basis is locally reducible but still LOCC indistinguishable.  If the recursion returns $\True$, the tree is both a certificate of distinguishability and an explicit LOCC protocol.

\begin{example}
\label{ex:locc-tree}
For the three-qubit matrix in Section~\ref{sec:example}, the root layer is the connected graph $K(P_1,N_1)$. Party $1$ measures in the basis $\{x_1,x_1^\perp\}$, splitting the candidate rows into $U_0=\{1,2,3,4\}$ and $U_1=\{5,6,7,8\}$. On either branch, the second induced layer is connected, so party $2$ measures in the basis $\{x_2,x_2^\perp\}$. This produces the four two-state subproblems $\{1,2\}$, $\{3,4\}$, $\{5,6\}$, and $\{7,8\}$. Party $3$ distinguishes the first two subproblems in the basis $\{x_3,x_3^\perp\}$ and the last two in the basis $\{x_4,x_4^\perp\}$. Every outcome sequence therefore ends at a singleton candidate set, and the recursion returns $\True$.

Figure~\ref{fig:locc-tree-example} displays the complete splitting tree. Reading it from the root, an internal-node label gives the current candidate set and the next local projective measurement, while the observed outcome selects one of the two outgoing edges. The entire tree is the adaptive three-round LOCC protocol; a single experimental run follows one root-to-leaf path and terminates at the row label of the unknown state.

\begin{figure}[H]
\centering
\resizebox{\linewidth}{!}{%
\begin{tikzpicture}[
  x=1cm,y=1cm,
  splitnode/.style={rounded corners=2pt, draw=black!65, line width=.45pt,
    align=center, inner sep=2.5pt, minimum height=7.5mm,
    font=\scriptsize},
  rootnode/.style={splitnode, fill=layerblue!10, minimum width=34mm},
  secondnode/.style={splitnode, fill=layerorange!12, minimum width=33mm},
  thirdnode/.style={splitnode, fill=layergreen!12, minimum width=28mm},
  leafnode/.style={rounded corners=1.5pt, draw=black!55, fill=black!3,
    inner sep=2pt, minimum width=10mm, font=\scriptsize},
  rootedge/.style={draw=layerblue, line width=.75pt},
  secondedge/.style={draw=layerorange, line width=.75pt},
  thirdedge/.style={draw=layergreen, line width=.75pt},
  outlabel/.style={font=\scriptsize, fill=white, inner sep=1pt}
]
\node[rootnode] (r) at (0,0)
  {$U=\{1,\ldots,8\}$\\party $1$: $\{x_1,x_1^\perp\}$};

\node[secondnode] (a) at (-4.15,-1.75)
  {$U=\{1,2,3,4\}$\\party $2$: $\{x_2,x_2^\perp\}$};
\node[secondnode] (b) at (4.15,-1.75)
  {$U=\{5,6,7,8\}$\\party $2$: $\{x_2,x_2^\perp\}$};

\node[thirdnode] (c) at (-6.25,-3.55)
  {$U=\{1,2\}$\\party $3$: $\{x_3,x_3^\perp\}$};
\node[thirdnode] (d) at (-2.05,-3.55)
  {$U=\{3,4\}$\\party $3$: $\{x_3,x_3^\perp\}$};
\node[thirdnode] (e) at (2.05,-3.55)
  {$U=\{5,6\}$\\party $3$: $\{x_4,x_4^\perp\}$};
\node[thirdnode] (f) at (6.25,-3.55)
  {$U=\{7,8\}$\\party $3$: $\{x_4,x_4^\perp\}$};

\foreach \x/\name/\lab in {
  -7.30/l1/1,-5.20/l2/2,-3.10/l3/3,-1.00/l4/4,
   1.00/l5/5, 3.10/l6/6, 5.20/l7/7, 7.30/l8/8}
  \node[leafnode] (\name) at (\x,-5.05) {$U=\{\lab\}$};

\draw[rootedge] (r)--node[outlabel, pos=.53, sloped] {$x_1$} (a);
\draw[rootedge] (r)--node[outlabel, pos=.53, sloped] {$x_1^\perp$} (b);
\draw[secondedge] (a)--node[outlabel, pos=.53, sloped] {$x_2$} (c);
\draw[secondedge] (a)--node[outlabel, pos=.53, sloped] {$x_2^\perp$} (d);
\draw[secondedge] (b)--node[outlabel, pos=.53, sloped] {$x_2$} (e);
\draw[secondedge] (b)--node[outlabel, pos=.53, sloped] {$x_2^\perp$} (f);
\draw[thirdedge] (c)--node[outlabel, pos=.54, sloped] {$x_3$} (l1);
\draw[thirdedge] (c)--node[outlabel, pos=.54, sloped] {$x_3^\perp$} (l2);
\draw[thirdedge] (d)--node[outlabel, pos=.54, sloped] {$x_3$} (l3);
\draw[thirdedge] (d)--node[outlabel, pos=.54, sloped] {$x_3^\perp$} (l4);
\draw[thirdedge] (e)--node[outlabel, pos=.54, sloped] {$x_4$} (l5);
\draw[thirdedge] (e)--node[outlabel, pos=.54, sloped] {$x_4^\perp$} (l6);
\draw[thirdedge] (f)--node[outlabel, pos=.54, sloped] {$x_4$} (l7);
\draw[thirdedge] (f)--node[outlabel, pos=.54, sloped] {$x_4^\perp$} (l8);
\end{tikzpicture}
}%
\caption{Complete color-splitting tree for Example~\ref{ex:locc-tree}. An internal node lists the active candidate set and the local projective measurement; an edge records the observed outcome and updates the candidate set. The singleton leaves identify the initial basis state.}
\label{fig:locc-tree-example}
\end{figure}
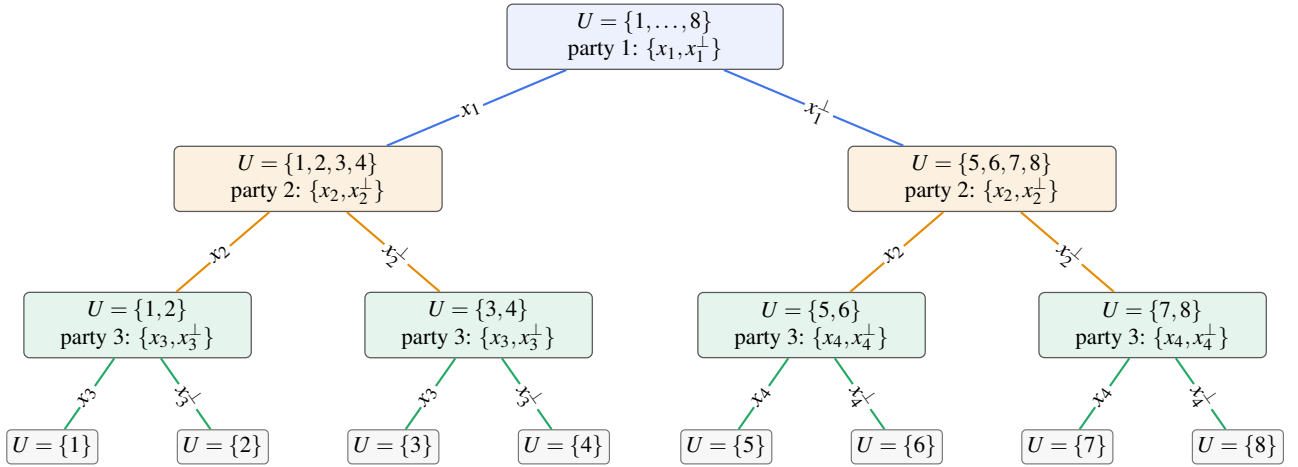
\end{example}
\FloatBarrier

\begin{example}
\label{ex:locc-indistinguishable}
The following three-qubit formal matrix gives a contrasting example:
\[
M_{\mathrm{irr}}=
\begin{pmatrix}
x_2          & x_3          & x_5          \\
x_1          & x_4          & x_5^\perp    \\
x_1          & x_3^\perp    & x_5          \\
x_1          & x_4^\perp    & x_5^\perp    \\
x_2^\perp    & x_3          & x_5          \\
x_1^\perp    & x_3          & x_5^\perp    \\
x_1^\perp    & x_3^\perp    & x_6          \\
x_1^\perp    & x_3^\perp    & x_6^\perp
\end{pmatrix}.
\]
Any two distinct rows contain a complementary pair in at least one column, so the rows define a complete three-qubit OPB.  With the vertices labeled by the row indices, its color layers are
\begin{align*}
\Gamma_1(M_{\mathrm{irr}})
 &=K(\{2,3,4\},\{6,7,8\})\sqcup K(\{1\},\{5\}),\\
\Gamma_2(M_{\mathrm{irr}})
 &=K(\{1,5,6\},\{3,7,8\})\sqcup K(\{2\},\{4\}),\\
\Gamma_3(M_{\mathrm{irr}})
 &=K(\{1,3,5\},\{2,4,6\})\sqcup K(\{7\},\{8\}).
\end{align*}
Thus every color layer is the disjoint union of $K_{3,3}$ and $K_{1,1}$ and is therefore disconnected.  Corollary~\ref{cor:local-irreducibility} shows directly that the OPB is locally irreducible and hence cannot be perfectly distinguished by any finite-round LOCC protocol.  In contrast to Example~\ref{ex:locc-tree}, no recursive search is needed here: the disconnected layers at the root already provide the certificate of LOCC indistinguishability.
\end{example}

\section{Bounds and asymptotic analysis}
\label{sec:bounds}
Write $a_n$ for the number of equivalence classes of $\mathcal O(n)$ under $\sim$; that is, the number of essentially different $n$-qubit OPBs.
\subsection{Counting single layers}
\label{sec:counting}
In this section we count single color layers. Suppose the $j$-th layer contains $k=|A_j|$ variables, and arrange the part sizes $d_\alpha$ in nonincreasing order to obtain the \emph{spectrum}
$\lambda=(d_1,\dots,d_k)$
of the layer. By Section~\ref{sec:not-mat}, $\lambda$ is a partition of $N/2$, written $\lambda\vdash N/2$:
\[
d_1\ge\dots\ge d_k\ge1,
\qquad
\sum_{i=1}^{k}d_i=\frac{N}{2}.
\]
\begin{definition}[Spectrum]
\label{def:spectrum}
A spectrum is a partition $\lambda=(d_1,\dots,d_k)\vdash N/2$. The quantity $m_s(\lambda)$ denotes the number of parts of $\lambda$ equal to $s$.
\end{definition}
\begin{theorem}[Counting single layers]
\label{thm:count}
Given the spectrum $\lambda=(d_1,\dots,d_k)\vdash N/2$ of one column, the number of \emph{labeled} color layers on the vertex set $V$ whose connected components have part sizes given by $\lambda$ is
\[
L_N(\lambda)
=
\frac{N!}
{2^k\prod_{i=1}^{k}(d_i!)^2\prod_{s\ge1}m_s(\lambda)!}.
\]
\end{theorem}
\begin{proof}
First count the ways of placing the vertex set $V$ into sets
\[
P_1,N_1,P_2,N_2,\dots,P_k,N_k
\]
of sizes
\[
d_1,d_1,d_2,d_2,\dots,d_k,d_k.
\]
There are
\[
\binom{N}{d_1,d_1,d_2,d_2,\dots,d_k,d_k}
=
\frac{N!}{\prod_{i=1}^{k}(d_i!)^2}
\]
such placements.
Since interchanging $P_i$ and $N_i$ does not matter, we may equivalently count the ways of placing $V$ into sets
$P_1\sqcup N_1,\dots,P_k\sqcup N_k$
of sizes $2d_1,\dots,2d_k$. There are
$\binom{N}{2d_1,2d_2,\dots,2d_k}$
such placements. Since interchanging blocks of equal size does not matter, this count becomes
\[
\binom{N}{2d_1,2d_2,\dots,2d_k}
\cdot
\frac{1}{\prod_s m_s(\lambda)!}.
\]
For each $P_i\sqcup N_i$ of size $2d_i$, splitting it into two unlabeled parts of size $d_i$ can be done in
$\frac12\binom{2d_i}{d_i}$
ways. Therefore
\[
\begin{aligned}
L_N(\lambda)
&=
\binom{N}{2d_1,2d_2,\dots,2d_k}
\cdot
\frac{1}{\prod_s m_s(\lambda)!}
\cdot
\prod_{i=1}^{k}
\left(
\frac12\binom{2d_i}{d_i}
\right)\\
&=
\frac{N!}
{2^k\prod_{i=1}^{k}(d_i!)^2\prod_s m_s(\lambda)!}.
\end{aligned}
\]
\end{proof}
\begin{definition}[Total number of single layers]
\label{def:LN}
The total number of single layers for a given $n$ with $N=2^n$ is $L_N=\sum_{\lambda\vdash N/2}L_N(\lambda)$.
\end{definition}
\begin{theorem}[Recurrence]
\label{thm:recur}
The numbers $L_N$ satisfy
\[
L_N
=
\sum_{d=1}^{N/2}
\binom{N-1}{2d-1}
\binom{2d-1}{d-1}
L_{N-2d},
\qquad
L_0=1.
\]
\end{theorem}
\begin{proof}
Fix a vertex $v\in V$, and suppose $v$ lies on the $P$-side of a complete bipartite block $K_{d,d}$. Here $1\le d\le N/2$. From the remaining $N-1$ vertices, choose $2d-1$ to form the block, then choose the $d-1$ of them that share the side of $v$. This gives
$$\binom{N-1}{2d-1}\binom{2d-1}{d-1}$$
choices. The remaining $N-2d$ vertices are partitioned in $L_{N-2d}$ ways. Summing over $d$ gives the recurrence.
\end{proof}
\begin{remark}
\label{rem:layers}
The first values are $L_0=1$, $L_2=1$, $L_4=6$, $L_6=70$, and $L_8=1365$, in agreement with a direct summation of Theorem~\ref{thm:count}; see Section~\ref{sec:asympt}.
\end{remark}
\subsection{Lower bound}
\label{sec:lb}
\begin{theorem}[Lower bound]
\label{thm:lb}
For every $n\ge1$,
\[
a_{n+1}\ge \binom{a_n+1}{2}.
\]
\end{theorem}
\begin{proof}
Take two $n$-qubit OPB formal matrices $A,B$ (possibly equal) with disjoint variable-name sets, and let $x_0,x_0^\perp$ be new variables not occurring in $A$ or $B$. Then
\[
\begin{pmatrix}
x_0 & A\\
x_0^\perp & B
\end{pmatrix}
\]
is a valid $(n+1)$-qubit OPB formal matrix. It has $2N=2^{n+1}$ rows; every row contains $n+1$ pairwise distinct letters; two rows within the same block are orthogonal by the orthogonality of $A$ or $B$, and two rows across blocks are orthogonal because their first entries are $x_0$ and $x_0^\perp$; realizability is evident.

Choosing two classes with repetition from the $a_n$ classes yields $\binom{a_n+1}{2}$ such $(n+1)$-qubit formal matrices, and these are pairwise inequivalent. Indeed, in the corresponding edge-colored complete multigraph, the first color layer is the unique layer containing a connected component whose two parts both have size $N$, namely $K_{N,N}$. All other layers come from $A$ or $B$, and their connected components have part sizes at most $N/2$. 

Hence any graph isomorphism must preserve this layer and map the remaining $n$ layers to the subgraphs on the two sides of $K_{N,N}$. By Theorem~\ref{thm:equiv}, this recovers exactly the equivalence classes of the unordered pair $\{A,B\}$, with $A=B$ allowed. Thus
$a_{n+1}\ge \binom{a_n+1}{2}$.
\end{proof}
\begin{corollary}[Exponential lower bound]
\label{cor:lb}
If $a_{n_0}\ge B$ for some $n_0$, then for all $n\ge n_0$,
\[
a_n
\ge
2^{1+(\log B-1)2^{n-n_0}}.
\]
In particular, $a_n\ge2^{\Omega(2^n)}$.
\end{corollary}
\begin{proof}
Set $b_{n_0}=B$ and
$b_{k+1}=\frac{b_k(b_k+1)}{2}$ for $k\ge n_0$.
By Theorem~\ref{thm:lb}, $a_k\ge b_k$. Since
$b_{k+1}\ge \frac{b_k^2}{2}$,
we obtain
$\log b_{k+1}\ge 2\log b_k-1$.
With $c_k=\log b_k$, iteration yields
\[
\begin{aligned}
c_n
&\ge
2^{n-n_0}c_{n_0}
-
(2^{n-n_0}-1)\\
&=
1+(\log B-1)2^{n-n_0}.
\end{aligned}
\]
Hence
$a_n \ge 2^{1+(\log B-1)2^{n-n_0}}$ for $n\ge n_0$.
When $B\ge2$, the exponent grows like $\Theta(2^n)$, so
$a_n\ge2^{\Omega(2^n)}$.
\end{proof}
\begin{remark*}
The known exact values $a_2=2$, $a_3=17$, and $a_4=27385$ may serve as the starting value $B$ \cite{chen-djok,chen-jiang}. The numerical table in Section~\ref{sec:asympt} uses $B=a_4=27385$.
\end{remark*}

\subsection{Upper bound}
\label{sec:ub}
\begin{theorem}[Upper bound]
\label{thm:ub}
For every $n\ge1$,
\[
a_n\le B_{2^{n-1}}^n,
\]
where $B_m$ denotes the Bell number of an $m$-element set, i.e., the number of all set partitions of an $m$-element labeled set.
\end{theorem}
\begin{proof}
Write 
\begin{eqnarray}
\label{eq:N=2^n,m=2^{n-1}}    
N=2^n, m=2^{n-1}.
\end{eqnarray}
For an OPB formal matrix $M$, the $j$-th color layer can be written as
\[
\Gamma_j(M)
=
\bigsqcup_{\alpha\in A_j}
K(P_\alpha,N_\alpha).
\]
By Lemma~\ref{lem:balance}, each connected component satisfies
$|P_\alpha|=|N_\alpha|=d_\alpha$,
and
$\sum_{\alpha\in A_j}d_\alpha=m$.
For each connected component of layer $j$, choose one of the two orientations, and define
\[
P_j=\bigsqcup_{\alpha\in A_j}P_\alpha,
\qquad
N_j=\bigsqcup_{\alpha\in A_j}N_\alpha.
\]
Then
\[
V=P_j\sqcup N_j,
\qquad
|P_j|=|N_j|=m.
\]
Merging all connected components of layer $j$ gives a balanced complete bipartite graph
$\widehat{\Gamma}_j=K(P_j,N_j)$
on $V$, called the \emph{coarsened layer} of color $j$. Clearly,
$E(\Gamma_j)\subseteq E(\widehat{\Gamma}_j)$.
Since the original $n$ color layers cover $K_N$, so do the coarsened layers:
\[
E(K_N)
=
\bigcup_{j=1}^{n}
E(\widehat{\Gamma}_j).
\]
For each vertex $r\in V$, define its \emph{side marking} with respect to all coarsened layers,
$\varepsilon(r) = (\varepsilon_1(r),\dots,\varepsilon_n(r)) \in\{0,1\}^n$,
where
\[
\varepsilon_j(r)
=
\begin{cases}
0,&r\in P_j,\\
1,&r\in N_j.
\end{cases}
\]
Take two distinct vertices $r,s\in V$. By the covering condition, there is a color $j$ with $\{r,s\}\in E(\Gamma_j)$. Then $r$ and $s$ lie on opposite sides of the component $K(P_\alpha,N_\alpha)$ containing them, hence on opposite sides of the coarsened layer $K(P_j,N_j)$. Therefore $\varepsilon_j(r)\ne\varepsilon_j(s)$, so
$\varepsilon(r)\ne\varepsilon(s)$.

Thus $\varepsilon:V\to\{0,1\}^n$ is injective. Since both sets have size $2^n$, it is a bijection.
Fix a color $j$. Since every side marking corresponds to exactly one vertex, for every $r\in V$ there is a unique vertex $\iota_j(r)$ whose marking $\varepsilon(\iota_j(r))$ differs from $\varepsilon(r)$ only in the $j$-th position. Moreover,
$\iota_j(\iota_j(r))=r$ and $\iota_j(r)\ne r$.
Hence these vertex pairs form a perfect matching of $V$:
\[
\mathcal{M}_j
=
\bigl\{\{r,\iota_j(r)\}:r\in V\bigr\},
\qquad
|\mathcal{M}_j|
=
\frac{N}{2}
=
m.
\]
We show that every edge of $\mathcal{M}_j$ must occur in the original $j$-th color layer. Take $\{r,\iota_j(r)\}\in\mathcal{M}_j$. Since all color layers cover $K_N$, this pair is adjacent in some layer $\Gamma_k$. If
$\{r,\iota_j(r)\}\in E(\Gamma_k)$,
then the two vertices lie on opposite sides of the coarsened layer $K(P_k,N_k)$, so their markings differ in the $k$-th position. By the definition of $\iota_j(r)$, however, the two markings differ only in the $j$-th position; hence $k=j$. Therefore
$\mathcal{M}_j\subseteq E(\Gamma_j)$.
We call $\mathcal{M}_j$ the \emph{forced perfect matching} of the $j$-th color layer.
For each $\alpha\in A_j$, define
\[
\mathcal{B}_\alpha
=
\{e\in\mathcal{M}_j:
e\subseteq P_\alpha\sqcup N_\alpha\}.
\]
The vertex sets of distinct variable branches are disjoint, so the $\mathcal{B}_\alpha$ are pairwise disjoint. Every forced matching edge lies in some connected component of layer $j$, so
\[
\mathcal{M}_j
=
\bigsqcup_{\alpha\in A_j}
\mathcal{B}_\alpha.
\]
Each $\mathcal{B}_\alpha$ is nonempty. Indeed, if $r\in P_\alpha\sqcup N_\alpha$, then $\{r,\iota_j(r)\}\in E(\Gamma_j)$; since distinct connected components of the same layer have no edges between them, $\iota_j(r)$ still belongs to $P_\alpha\sqcup N_\alpha$. Thus the variable branches of layer $j$ give a set partition of the forced matching edges $\mathcal{M}_j$.

Conversely, fix a coarsened layer $K(P_j,N_j)$ with an orientation, and take a set partition
$\mathcal{M}_j = \mathcal{B}_1\sqcup\dots\sqcup\mathcal{B}_k$
of the forced matching. For each block $\mathcal{B}_t$, let $P(\mathcal{B}_t)$ and $N(\mathcal{B}_t)$ be the sets of endpoints of its matching edges lying in $P_j$ and $N_j$, respectively. Each matching edge has one endpoint on each side, so
$|P(\mathcal{B}_t)| = |N(\mathcal{B}_t)| = |\mathcal{B}_t|$.
The block therefore determines a balanced complete bipartite graph
$K(P(\mathcal{B}_t),N(\mathcal{B}_t))$.

The disjoint union of these graphs recovers a possible $j$-th color layer. Hence, for a fixed coarsened layer, the variable-branch structures of layer $j$ are in bijection with the set partitions of $\mathcal{M}_j$.

The matching $\mathcal{M}_j$ consists of $m=2^{n-1}$ labeled edges, so the number of its set partitions is the Bell number $B_m$. Moreover, by the bijectivity of the side-marking map, any two systems of coarsened layers are isomorphic under a vertex map that preserves the markings, so no additional factor for the choice of coarsened layers is needed under graph equivalence.

Temporarily regard the $n$ colors as labeled and ignore whether combinations of partitions of the different layers still cover $K_N$. Then all $n$ layers admit at most
$B_m^n$
combinations. Illegal partition combinations and equivalences coming from vertex or color permutations can only decrease the count. Therefore
\[
a_n \le B_m^n = B_{2^{n-1}}^n.\]
\end{proof}
\begin{remark}
\label{rem:bell}
The Bell number $B_m$ is the number of all set partitions of an $m$-element labeled set. If
$\genfrac\{\}{0pt}{}{m}{k}$
denotes the Stirling number of the second kind, i.e., the number of partitions of an $m$-element set into exactly $k$ nonempty blocks, then
$B_m = \sum_{k=0}^{m} \genfrac\{\}{0pt}{}{m}{k}$.
The Bell numbers also satisfy the recurrence
$B_{m+1} = \sum_{r=0}^{m} \binom{m}{r} B_{m-r}$.
\end{remark}

\begin{example}[Forced matchings for $n=3$]
\label{ex:forced}
Take $n=3$, so $N=8$ and $m=4$ in \eqref{eq:N=2^n,m=2^{n-1}}. Label the eight vertices by their side markings in the three coarsened layers as the binary strings $000,001,010,011,100,101,110,111$. The binary string records on which side of each of the three coarsened balanced cuts the vertex lies.
Consider the third color layer. The vertex differing from a given vertex only in the third side marking is unique, so the forced perfect matching of the third layer is
\[
\mathcal{M}_3
=
\bigl\{
\{000,001\},
\{010,011\},
\{100,101\},
\{110,111\}
\bigr\}.
\]
Suppose the first two matching edges form one block and the last two another:
\[
\begin{gathered}
\mathcal{B}_1
=
\bigl\{
\{000,001\},
\{010,011\}
\bigr\},\\
\mathcal{B}_2
=
\bigl\{
\{100,101\},
\{110,111\}
\bigr\}.
\end{gathered}
\]
Block $\mathcal{B}_1$ gives the endpoint sets $\{000,010\}$ and $\{001,011\}$ on the two sides of the third coarsened layer, so it determines the variable branch
$K(\{000,010\},\{001,011\})$.
Similarly, $\mathcal{B}_2$ determines
$K(\{100,110\},\{101,111\})$.
Hence the third color layer is
$\Gamma_3 = K(\{000,010\},\{001,011\}) \sqcup K(\{100,110\},\{101,111\})$.
Placing several forced matching edges into the same block is equivalent to letting the endpoints of these edges use the same variable pair in that layer. Since $\mathcal{M}_3$ has four labeled matching edges, it admits
$B_4=15$
set partitions. 

Thus, for a fixed coarsened layer, the third color layer admits at most $15$ variable-branch structures. These side markings are realized by the three-qubit OPB of Section~\ref{sec:example}: its rows $1,\dots,8$ carry the markings $000,\dots,111$ in this order, and the partition $\mathcal{B}_1\sqcup\mathcal{B}_2$ above yields exactly the third color layer of that OPB.
\end{example}
\subsection{Asymptotic analysis and numerical results}
\label{sec:asympt}
In this section, $\log$ denotes the base-$2$ logarithm and $\ln$ the natural logarithm.
\begin{theorem}[Asymptotic behavior]
\label{thm:asympt}
There exist constants $c,C>0$ such that
\[
2^{c2^n}
\le
a_n
\le
2^{Cn^22^n}
\]
for all sufficiently large $n$. More precisely,
\[
a_n
=
2^{2^{n+o(n)}},
\qquad
n\to\infty.
\]
\end{theorem}
\begin{proof}
The Bell numbers satisfy \cite{debruijn}
\[
\ln B_m
=
m\ln m
-
m\ln\ln m
-
m
+
o(m).
\]
Thus, with $m=2^{n-1}$ and in base-$2$ logarithms,
\[
\log B_{2^{n-1}}
=
2^{n-1}
\left[
n-\log n-1-\log(e\ln2)+o(1)
\right],
\]
where
$\log(e\ln2)=0.91393\dots$.
By Theorem~\ref{thm:ub} and Corollary~\ref{cor:lb},
\[
c\,2^n
\le
\log a_n
\le
n\log B_{2^{n-1}}
=
O(n^22^n).
\]
Equivalently,
\[
2^{c2^n}
\le
a_n
\le
2^{Cn^22^n}
\]
for constants $c,C>0$. Taking $\log$ once more gives
\[
n+O(1)
\le
\log\log a_n
\le
n+2\log n+O(1).
\]
Dividing by $n$ yields
\[
\lim_{n\to\infty}
\frac{\log\log a_n}{n}
=
1,
\]
since $(\log n)/n\to0$. Thus
\[
\log\log a_n
=
n+o(n),
\]
which is equivalent to
\[
a_n
=
2^{2^{n+o(n)}}.
\]
\end{proof}
\begin{table}[t]
\centering
\small
\caption{Numerical results for $n=2,\dots,8$.}
\label{tab:numerics}
\begin{tabular}{ccccc}
\toprule
$n$ & $N=2^n$ & $L_N$ & lower bound & upper bound $B_{2^{n-1}}^n$ \\
\midrule
2 & 4 & 6 & 2 (exact) & 4 \\
3 & 8 & 1365 & 17 (exact) & 3375 \\
4 & 16 & $4.70359\times10^9$ & 27385 (exact) & $2.93766\times10^{14}$ \\
5 & 32 & $3.25125\times10^{26}$ & $3.74983\times10^8$ & $1.26426\times10^{50}$ \\
6 & 64 & $8.84292\times10^{67}$ & $7.03061\times10^{16}$ & $4.41140\times10^{156}$ \\
7 & 128 & $5.06137\times10^{166}$ & $2.47147\times10^{33}$ & $4.47785\times10^{456}$ \\
8 & 256 & $4.32579\times10^{396}$ & $3.05408\times10^{66}$ & $2.51267\times10^{1264}$ \\
\bottomrule
\end{tabular}
\end{table}
Table~\ref{tab:numerics} gives the numerical results for $n=2,\dots,8$. The total number of single layers $L_N$ is computed from Theorem~\ref{thm:count} and the recurrence of Theorem~\ref{thm:recur}. The \emph{lower bound} is obtained by iterating Theorem~\ref{thm:lb} starting from the exact value $a_4=27385$ \cite{chen-djok,chen-jiang}; the values for $n=2,3,4$ are exact. The \emph{upper bound} $B_{2^{n-1}}^n$ comes from Theorem~\ref{thm:ub}.
\begin{remark*}
All values in Table~\ref{tab:numerics} were independently verified using the recurrence of Theorem~\ref{thm:recur} and the Bell-number recurrence of Remark~\ref{rem:bell}.
\end{remark*}
\section{Graph algorithms for equivalence and LOCC discrimination}
\label{sec:alg}
The multigraph representation gives two decision procedures.  The first tests equivalence of OPB formal matrices by combining inexpensive structural invariants with an existing exact graph-isomorphism algorithm.  The second evaluates the connectivity and recursive splitting criteria of Section~\ref{sec:local-locc}; when the recursion succeeds, it also records the complete sequence of local projective measurements needed for finite-round LOCC discrimination.

\subsection{Algorithm I: testing equivalence}
\label{sec:alg-overview}
By Theorem~\ref{thm:equiv}, deciding whether two OPB formal matrices are equivalent is exactly the isomorphism problem for their edge-colored complete multigraphs, with both vertex and color permutations allowed.  The algorithm has two phases.  \emph{Phase 1} compares necessary invariants and rejects pairs that are immediately seen to be inequivalent.  If every invariant agrees, \emph{Phase 2} calls an existing exact graph-isomorphism algorithm on the two associated multigraphs \cite{mckay,babai,cordella}.
The input consists of two OPB formal matrices $M_1,M_2\in \mathcal O(n)$; the output is \True{} if and only if $M_1\sim M_2$. We call $M_1$ the \emph{source} matrix and $M_2$ the \emph{target} matrix.
\subsubsection{Decomposition}
\label{sec:alg-decompose}
A decomposition procedure $\Decompose(M)$ collects, for each variable $\alpha$, its column $\chi(\alpha)$ and its two sides $P_\alpha,N_\alpha$, verifying that each variable occurs in exactly one column and that
$|P_\alpha|=|N_\alpha|>0$.
Otherwise, the input is rejected as not a valid OPB formal matrix. Each variable $\alpha$ forms a \emph{component} $K(P_\alpha,N_\alpha)$ of size
$d_\alpha=|P_\alpha|$.
Each column $j$ is a \emph{layer} consisting of all its components, with the \emph{column spectrum}
$\pi_j(M) = \sort_{\ge} \{d_\alpha:\alpha\in A_j\}$.
The recovered layers determine both the invariants used in Phase 1 and the multigraphs passed to Phase 2.
\subsubsection{Phase 1: filtering by necessary invariants}
\label{sec:alg-phase1}
We compare the following invariants; if any of them differs, the two matrices are necessarily inequivalent. Equality of all invariants is not sufficient, so Phase 2 is still required.
\begin{enumerate}
\item \emph{Size}: $N$ and $n$.
\item \emph{Number of variables}: $v$.
\item \emph{Global spectrum}:
$\pi(M) = \sort_{\ge}(d_1,\dots,d_v)$.
\item \emph{Multiset of column spectra}:
$\{\pi_j(M):1\le j\le n\}$.
Column permutations are allowed, so this is compared as a multiset.
\item \emph{Multiset of edge multiplicities}: for each pair of rows $r<s$, write
\begin{multline*}
\mu(r,s)
=
\#\{j:
r,s\text{ lie in the same block}
\\
\text{ of column }j
\text{ with opposite directions}\},
\end{multline*}
and compare the multisets
$\{\mu(r,s):1\le r<s\le N\}$.
\item \emph{Multiset of vertex degrees}: define
$\deg(r) = \sum_{s\ne r}\mu(r,s)$,
and compare the multisets
$\{\deg(r):r\in V\}$.
\end{enumerate}
Items~(5) and~(6) are precisely the statistics of the edge multiplicities and vertex degrees of the edge-colored complete multigraph, which are invariant under isomorphism. All conditions above are necessary:
$M_1\sim M_2 \implies \Inv(M_1)=\Inv(M_2)$,
but the converse need not hold.
\subsubsection{Phase 2: graph-isomorphism test}
\label{sec:alg-phase2}
If Phase 1 is passed, apply any existing exact graph-isomorphism algorithm to $\Gamma(M_1)$ and $\Gamma(M_2)$, allowing simultaneous permutations of the vertices and the colors.  Return \True{} precisely when the graph-isomorphism algorithm declares the two multigraphs isomorphic; otherwise return \False.

\subsubsection{Correctness}
\label{sec:alg-correct}
\begin{theorem}[Correctness]
\label{thm:alg}
For any $M_1,M_2\in \mathcal O(n)$, Algorithm I outputs \True{} if and only if $M_1\sim M_2$.
\end{theorem}
\begin{proof}
Every quantity compared in Phase 1 is invariant under equivalence, so this phase never rejects an equivalent pair.  If all invariants agree, the exact graph-isomorphism algorithm used in Phase 2 returns \True{} exactly when $\Gamma(M_1)$ and $\Gamma(M_2)$ are isomorphic.  Theorem~\ref{thm:equiv} identifies this condition with $M_1\sim M_2$.
\end{proof}
\subsection{Algorithm II: deciding finite-round LOCC distinguishability}
\label{sec:alg-locc}
The input is a standard-form realization $\mathcal B(M)$, and the algorithm uses the color layers of $\Gamma(M)$ as follows.
\begin{enumerate}
\item If $\Gamma_j(M)$ is disconnected for every $j$, return that $\mathcal B(M)$ is locally irreducible and hence cannot be perfectly distinguished by finite-round LOCC.  If at least one root layer is connected, continue to the recursive test.
\item Evaluate the recursion $\mathsf D(U,J)$ in \eqref{eq:dist-recursion}, beginning at $(V,\{1,\ldots,n\})$.  At a state $(U,J)$ with $|U|>1$, examine every $j\in J$ for which $\Gamma_j(M)[U]$ is connected.  Such a color determines the unique variable $\alpha(j,U)$ and the two nonempty candidate sets
\[
U_0=U\cap P_{\alpha(j,U)}(M),
\qquad
U_1=U\cap N_{\alpha(j,U)}(M).
\]
Recursively test $(U_0,J\setminus\{j\})$ and $(U_1,J\setminus\{j\})$.  If both calls return \True, record $j$, $\alpha(j,U)$, and the two child states at the current node, and return \True.  If no color succeeds, return \False.  Memoizing each state $(U,J)$ avoids reevaluating the same subproblem.
\end{enumerate}
If the root returns \False, Theorem~\ref{thm:locc-tree} implies that the basis is not perfectly distinguishable by finite-round LOCC.  This includes the locally reducible case in which a connected root layer exists but every recursive continuation fails.

If the root returns \True, the recorded choices form a complete color-splitting tree and specify the LOCC protocol without any further search.  At a recorded node $(U,J)$ with choice $(j,\alpha(j,U))$, party $j$ performs the projective measurement
\[
\left\{
|u_{\alpha(j,U)}\rangle\langle u_{\alpha(j,U)}|,
|u_{\alpha(j,U)}^\perp\rangle\langle u_{\alpha(j,U)}^\perp|
\right\}
\]
and broadcasts the outcome.  The first outcome replaces the active candidates by $U_0$ and the second by $U_1$; the parties then follow the corresponding recorded child with party $j$ removed from $J$.  Every branch terminates at a singleton leaf, whose row label identifies the initial product state.  Since one party is removed at each internal node, the protocol has depth at most $n$.

The correctness of Algorithm II follows from Corollary~\ref{cor:local-irreducibility} and Theorem~\ref{thm:locc-tree}.  With memoization, every state $(U,J)$ reached by the recursion is evaluated once.  If $R(M)$ denotes the number of such states, a crude bound obtained by checking all remaining induced color layers at each state is $O(R(M)nN^2)$ time and $O(R(M)N)$ space.  A successful recorded tree has at most $2N-1$ nodes.
\section{Conclusion}
\label{sec:concl}
We have associated to every $n$-qubit OPB formal matrix an edge-colored complete multigraph on the vertex set
$V=\{1,\dots,2^n\}$,
in which each variable pair contributes a balanced complete bipartite graph, each column contributes a color layer, and the $n$ color layers cover $K_{2^n}$. On this basis, we proved the equivalence theorem,
$M_1\sim M_2 \iff \Gamma(M_1)\cong\Gamma(M_2)$,
thereby reducing the classification of OPBs to a graph-isomorphism problem.
The same representation also has an operational use. For standard-form realizations, we proved that local irreducibility is equivalent to disconnection of every color layer and that perfect finite-round LOCC distinguishability is equivalent to the existence of a complete color-splitting tree. Whenever the latter condition holds, the tree directly specifies a protocol of local projective measurements.

We also proved the variable bound
$v\le2^n-1$
via subcube partitions and Tarsi's lemma, gave the exact count $L_N(\lambda)$ of single layers with prescribed spectrum together with its recurrence, and established
\[
a_{n+1}
\ge
\binom{a_n+1}{2},
\qquad
a_n
\le
B_{2^{n-1}}^n,
\]
with the resulting asymptotic behavior
\[
a_n
=
2^{2^{n+o(n)}},
\qquad
n\to\infty.
\]
Moreover, we developed an exact algorithm for testing equivalence by combining a filter of necessary invariants with an existing exact graph-isomorphism algorithm.  The same graph representation gives a constructive LOCC procedure: the recursive test either certifies that no complete color-splitting tree exists or records a tree whose nodes specify the successive local projective measurements and their conditional branches.

Two natural problems remain open. First, it would be interesting to obtain tighter lower and upper bounds on $a_n$, thereby narrowing the gap between the current estimates and achieving a more precise understanding of the growth of the number of equivalence classes. Second, it remains to find an effective algorithm that enumerates all equivalence classes of $n$-qubit OPBs for a given $n$. Such an algorithm would make it possible to obtain exact classifications for larger values of $n$ and may also reveal further structural properties of multiqubit OPBs.


\begin{thebibliography}{17}
\bibitem{ben-nonloc}
C.~H. Bennett, D.~P. DiVincenzo, C.~A. Fuchs, T.~Mor, E.~Rains, P.~W. Shor, J.~A. Smolin, and W.~K. Wootters,
\emph{Quantum nonlocality without entanglement},
Phys. Rev. A \textbf{59} (1999), 1070--1091.
DOI: 10.1103/PhysRevA.59.1070.
\bibitem{feng-shil}
Y.~Feng and Y.~Shi,
\emph{Characterizing locally indistinguishable orthogonal product states},
IEEE Trans. Inf. Theory \textbf{55} (2009), 2799--2806.
DOI: 10.1109/TIT.2009.2018330.
\bibitem{lebl}
J.~Lebl, A.~Shakeel, and N.~Wallach,
\emph{Local distinguishability of generic unentangled orthonormal bases},
Phys. Rev. A \textbf{93} (2016), 012330.
DOI: 10.1103/PhysRevA.93.012330; arXiv:1502.06639.
\bibitem{chen-jiang}
L.~Chen and Y.~Jiang,
\emph{Nonlocality via multiqubit orthogonal product bases},
Phys. Scr. \textbf{99} (2024), 065113.
DOI: 10.1088/1402-4896/ad46c7.
\bibitem{ben-unext}
C.~H. Bennett, D.~P. DiVincenzo, T.~Mor, P.~W. Shor, J.~A. Smolin, and B.~M. Terhal,
\emph{Unextendible product bases and bound entanglement},
Phys. Rev. Lett. \textbf{82} (1999), 5385--5388.
DOI: 10.1103/PhysRevLett.82.5385.
\bibitem{dmss}
D.~P. DiVincenzo, T.~Mor, P.~W. Shor, J.~A. Smolin, and B.~M. Terhal,
\emph{Unextendible product bases, uncompletable product bases and bound entanglement},
Commun. Math. Phys. \textbf{238} (2003), 379--410.
DOI: 10.1007/s00220-003-0877-6.
\bibitem{shi-graph}
F.~Shi, G.~Bai, X.~Zhang, Q.~Zhao, and G.~Chiribella,
\emph{Graph-theoretic characterization of unextendible product bases},
Phys. Rev. Research \textbf{5} (2023), 033144.
DOI: 10.1103/PhysRevResearch.5.033144.
\bibitem{xu-iso}
G.-B.~Xu, Y.-Y.~Zhu, D.-H.~Jiang, and Y.-G.~Yang,
\emph{Isomorphism of nonlocal sets of orthogonal product states in bipartite quantum systems},
Physica A \textbf{619} (2023), 128734.
DOI: 10.1016/j.physa.2023.128734.
\bibitem{chen-djok}
L.~Chen and D.~\v{Z}. \DJ{}okovi\'c,
\emph{Orthogonal product bases of four qubits},
J. Phys. A: Math. Theor. \textbf{50} (2017), 395301.
DOI: 10.1088/1751-8121/aa8546; arXiv:1606.06254.
\bibitem{aharoni}
R.~Aharoni and N.~Linial,
\emph{Minimal non-two-colorable hypergraphs and minimal unsatisfiable formulas},
J. Comb. Theory Ser. A \textbf{43} (1986), 196--204.
DOI: 10.1016/0097-3165(86)90060-9.
\bibitem{firmus}
Y.~Filmus, E.~A. Hirsch, S.~Kurz, F.~Ihringer, A.~Riazanov, A.~V. Smal, and M.~Vinyals,
\emph{Irreducible subcube partitions},
Electron. J. Combin. \textbf{30} (2023), P3.29.
DOI: 10.37236/11862.
\bibitem{debruijn}
N.~G. de~Bruijn,
\emph{Asymptotic Methods in Analysis},
2nd ed., North-Holland, Amsterdam, 1961, pp.~104--108.
\bibitem{mckay}
B.~D. McKay and A.~Piperno,
\emph{Practical graph isomorphism, II},
J. Symb. Comput. \textbf{60} (2014), 94--112.
DOI: 10.1016/j.jsc.2013.09.003.
\bibitem{babai}
L.~Babai,
\emph{Graph isomorphism in quasipolynomial time},
in: Proceedings of the 48th Annual ACM Symposium on Theory of Computing (STOC), 2016, pp.~684--697.
DOI: 10.1145/2897518.2897542.
\bibitem{cordella}
L.~P. Cordella, P.~Foggia, C.~Sansone, and M.~Vento,
\emph{A (sub)graph isomorphism algorithm for matching large graphs},
IEEE Trans. Pattern Anal. Mach. Intell. \textbf{26} (2004), 1367--1372.
DOI: 10.1109/TPAMI.2004.75.
\bibitem{de-rinaldis}
S.~De Rinaldis,
\emph{Distinguishability of complete and unextendible product bases},
Phys. Rev. A \textbf{70} (2004), 022309.
DOI: 10.1103/PhysRevA.70.022309; arXiv:quant-ph/0304027.
\end{thebibliography}
\end{document}